%% file: report.tex
\documentclass[a4paper,conference]{IEEEtran}

\usepackage{ifthen}
\newboolean{arxiv}
\setboolean{arxiv}{true}

\usepackage{amsmath,amssymb}
\usepackage{amsthm} 
\usepackage{ascmac}
\usepackage{bbm}
\usepackage{bm} 
\usepackage{braket}
\usepackage{cite}
\usepackage{dcolumn} 
\usepackage{enumerate}
\usepackage{enumitem}
\usepackage{framed}
\usepackage{graphicx} 
\ifthenelse{\boolean{arxiv}}{}{\usepackage[dvipdfmx]{hyperref}}
\usepackage{ifthen}
\usepackage{longtable}
\usepackage{mathtools}
\usepackage{multirow}
\usepackage{overpic}
\usepackage{pict2e}
\ifthenelse{\boolean{arxiv}}{}{\usepackage{pxjahyper}}
\usepackage{txfonts} 
\usepackage{xparse}

\usepackage{color}
\usepackage[most]{tcolorbox}

\newcommand{\hypercolor}{blue}
\hypersetup{
  colorlinks,
  citecolor=\hypercolor,
  linkcolor=\hypercolor,
  urlcolor=\hypercolor,
  bookmarksopen=true,
  bookmarksopenlevel=4,
  bookmarksnumbered
}
\ifthenelse{\boolean{arxiv}}{}{\mathtoolsset{showonlyrefs,showmanualtags}}

\input{settings_quant.tex}

\newcommand{\AB}{{A \otimes B}}

\newcommand{\seq}{\mathrm{SEQ}}
\newcommand{\Meas}{\mathbf{Meas}}
\newcommand{\Prep}{\mathbf{Prep}}
\newcommand{\Measseq}{\Meas^\seq}
\newcommand{\MeasseqAB}{\Measseq_{A \to B}}
\newcommand{\MeasLOCC}{\Meas^\mathrm{LOCC}}
\newcommand{\MeasSEP}{\Meas^\mathrm{SEP}}
\newcommand{\MeasPT}{\Meas^\mathrm{PT}}
\newcommand{\Measgen}{\mM}

\renewcommand{\ident}{\mathbbm{1}}
\newcommand{\id}{\mathrm{id}}
\newcommand{\gdis}{\raisebox{-.1em}{\includegraphics[scale=0.5]{figures/text_discard.pdf}}}
\newcommand{\gdisT}{\raisebox{-.1em}{\includegraphics[scale=0.5]{figures/text_discardT.pdf}}}
\newcommand{\cross}{\times}

\newcommand{\N}{\mathrm{N}}
\renewcommand{\P}{\mathrm{P}}
\newcommand{\F}{\mathrm{F}}

\newcommand{\Proc}{\mathbf{Proc}}
\newcommand{\ProcF}{\mathbf{Proc}^\F}
\newcommand{\ProcD}{\mathbf{Proc}^\mathrm{D}}
\newcommand{\St}{\mathbf{St}}
\newcommand{\StP}{\mathbf{St}^\P}
\newcommand{\StF}{\mathbf{St}^\F}
\newcommand{\StN}{\mathbf{St}^\N}
\newcommand{\StNP}{\mathbf{St}^{\N\P}}
\newcommand{\Eff}{\mathbf{Eff}}

\newcommand{\EffF}{\mathbf{Eff}^\F}
\newcommand{\Scalar}{\mathbf{Scalar}}
\newcommand{\ScalarF}{\mathbf{Scalar}^\F}
\renewcommand{\Vec}{\mathbf{V}}

\newcommand{\NA}{{N_A}}
\newcommand{\NB}{{N_B}}
\newcommand{\NC}{{N_C}}

\newcommand{\ol}{\overline}
\newcommand{\eqlocal}{\overset{\mathrm{local}}{=}}

\renewcommand{\PS}{P}

\newcommand{\Aut}{\mathrm{Aut}}
\newcommand{\sym}{\diamondsuit}

\newcommand{\termdef}[1]{\textit{#1}}

\newboolean{figure}
\setboolean{figure}{true}
\newcommand{\InsertPDF}[1]{\iffigure\includegraphics[scale=1.0]{#1}\fi}
\ifthenelse{\boolean{figure}}{}{\renewenvironment{overpic}[2][]{}{}}

\newcommand{\craket}[1]{\mathinner{({#1})}}
\newcommand{\cra}[1]{{\mathinner{({#1}|}}}
\newcommand{\cket}[1]{{\mathinner{|{#1})}}}
\let\protect\relax
{\catcode`\|=\active
  \xdef\Craket{\protect\expandafter\noexpand\csname Craket \endcsname}
  \expandafter\gdef\csname Craket \endcsname#1{\begingroup
     \ifx\SavedDoubleVert\relax
       \let\SavedDoubleVert\|\let\|\BraDoubleVert
     \fi
     \mathcode`\|32768\let|\BraVert
     \left({#1}\right)\endgroup}
}

\definecolor{gray}{RGB}{128,128,128}

\newcommand{\Discard}[1]{}
         
\begin{document}
\title{Discrimination of symmetric states in operational probabilistic theory}
\author{
 \IEEEauthorblockN{Kenji~Nakahira}
 \IEEEauthorblockA{
  Quantum Information Science Research Center, \\
  Quantum ICT Research Institute, Tamagawa University \\
  6-1-1 Tamagawa-gakuen, Machida, Tokyo 194-8610 Japan\\
  {\footnotesize\tt E-mail: nakahira@lab.tamagawa.ac.jp} 
  \vspace*{-2.64ex}}
}

\maketitle

\begin{abstract}
 A state discrimination problem in an operational probabilistic theory (OPT)
 is investigated in diagrammatic terms.
 It is well-known that, in the case of quantum theory, if a state set has a certain symmetry,
 then there exists a minimum-error measurement having the same type of symmetry.
 However, to our knowledge, it is not yet clear whether this property also holds
 in a more general OPT.
 We show that it also holds in OPTs, i.e., for a symmetric state set,
 there exists a minimum-error measurement that has the same type of symmetry.
 It is also shown that this result can be utilized to optimize over
 a restricted class of measurements, such as sequential or separable measurements.
\end{abstract}

\IEEEpeerreviewmaketitle

\section{Introduction}

Operational probabilistic theories (OPTs) and other similar theories,
such as generalized probabilistic theories, provide a general operational framework
that allows us to better understand the physical structure of quantum theory
\cite{Bar-2007,Bar-Bar-Cla-Lei-2010,Chi-Dar-Per-2010,Chi-Dar-Per-2011,Gog-Sca-2017}.
OPTs can be interpreted as a generalization of probability theory,
including classical probability theory, quantum theory, and many others
(such as the theory of Popescu-Rohrlich boxes \cite{Pop-Roh-1994}).
One of the motivations for using OPTs is to investigate quantum processes
from an operational point of view, which helps us to deeply understand quantum theory.
Another motivation is that an OPT might be useful in developing new physical
theories, such as a theory of quantum gravity.

One of the fundamental problems in probability theory is the state discrimination.
In the case of quantum theory, a vast number of studies have been carried out to
obtain an optimal measurement with respect to some criteria
(e.g., \cite{Hol-1973,Yue-Ken-Lax-1975,Hel-1976,Eld-Meg-Ver-2003,Bel-1975,Ban-Kur-Mom-Hir-1997,Usu-Tak-Hat-Hir-1999,Eld-For-2001,Bar-Cro-2009,Nak-Kat-Usu-2015-general}).
Although obtaining a closed-form analytical solution for an optimal measurement
is generally very difficult,
it is known that if a state set has a certain symmetry,
then there exists an optimal measurement having the same type of symmetry
(e.g.,\cite{Bel-1975,Ban-Kur-Mom-Hir-1997,Usu-Tak-Hat-Hir-1999,Eld-For-2001,Usu-Usa-Tak-Hat-2002,And-Bar-Gil-Hun-2002,Kat-Hir-2003,Cho-Hsu-2003,Eld-Meg-Ver-2004,Saw-Tsu-Usu-2007,Car-Vig-2010,Nak-Usu-2012-self,Nak-Usu-2013-group}).
This property allows us to simplify finding an optimal measurement for a symmetric state set
analytically and/or numerically.
However, to our knowledge, this result has not been extended to a more general OPT.
Note that it would not be surprising if this result does not hold in general
since the space of states and that of effects are not symmetric
in the case of a general system of OPTs, while they are highly symmetric
in the case of any quantum system.

In this paper, we investigate a state discrimination problem in an OPT;
we consider the case in which a state set has a certain symmetry,
in which case we show that there exists an optimal measurement that has the same type of symmetry.
This result can be proved without reference to specific algebraic structures
such as Hilbert spaces and operator algebras.
We also show that this result can be applied to the discrimination problem
over a restricted class of measurements.
As examples, we discuss four classes of measurements: sequential,
local operations and classical communication (LOCC), separable,
and partially transformable (PT).
It is worth noting that, in this paper,
we restrict our attention to the minimum-error strategy to simplify the discussion.
However, the results given in this paper can be easily applied to
other various criteria (see \cite{Nak-Kat-Usu-2015-general,Nak-Kat-Usu-2018-seq_gen}) in OPTs.

\section{Brief summary of operational probabilistic theories (OPTs)}

In this section, we briefly review the framework of OPTs.
The framework can be explained in several ways, leading to essentially almost the same formalism.
The proofs of some of the results are not presented in this paper,
which can be found in, e.g., Refs.~\cite{Bar-2007,Chi-Dar-Per-2011,Har-2011,Jan-Lal-2013,Nak-2019}.
Note that we consider only fixed causal structure.
We use diagrammatic representations that are used in Ref.~\cite{Nak-2019}
to represent formulae in an intuitive way,
which is motivated by the work of Coecke, Abramsky,
and others (see, e.g., \cite{Coe-2003,Abr-Coe-2004,Coe-Kis-2017}).

\subsection{Systems and processes}

An OPT consists of a collection of \termdef{systems} and a collection of \termdef{processes}.
Systems and processes respectively represent a physical system (e.g., a photon)
and a particular behavior of a physical process (e.g., a beam splitter).
Each process has input and output systems.
Let $\Proc_{A \to B}$ be the set of all processes having an input system $A$
and an output system $B$, referred to as processes from $A$ to $B$.
A \termdef{trivial} (or empty) system, denoted by $I$, is a special system.
A process from $I$ to $A$, denoted like $\cket{\rho}$, is called a \termdef{state} of $A$.
Similarly, a process from $A$ to $I$, denoted like $\cra{e}$, is
called an \termdef{effect} of $A$.
$\St_A \coloneqq \Proc_{I \to A}$ and $\Eff_A \coloneqq \Proc_{A \to I}$
are, respectively, called the \termdef{state space} and \termdef{effect space} of system $A$.
A process from $I$ to $I$ is called a \termdef{scalar}.
Let $\Scalar \coloneqq \Proc_{I \to I}$.

In diagrammatic terms,
a process $f \in \Proc_{A \to B}$, a state $\cket{\rho} \in \St_A$,
an effect $\cra{e} \in \Eff_A$, and a scalar $p \in \Scalar$ are depicted as
\begin{alignat}{1}
 \InsertPDF{figures/component.pdf} ~\raisebox{.1em}{.}
\end{alignat}
Labeled wires (labels are often omitted) represent systems, while boxes represent processes.
Each process has an input wire at the bottom and an output wire at the top.
$I$ is represented by `no wire'.
For a scalar, the box will be omitted.
Diagrammatic representations can be interpreted such as data flow diagrams,
where time increases from the bottom to the top.

\begin{ex}[(fully) quantum theory]
 For simplicity, we consider only finite-dimensional systems in the examples of quantum theory.
 Let $\Complex$ be the set of all complex numbers.
 Also, let $\mS(\Complex^n)$ and $\mS_+(\Complex^n)$ be, respectively, the sets of
 all complex Hermitian matrices and all complex positive semidefinite matrices of order $n$.
 $\St_A$ and $\Eff_A$ are isomorphic to $\mS_+(\Complex^\NA)$,
 where $\NA$ is a natural number determined by a system $A$.
 Note that $\mS_+(\Complex^\NA)$ is a symmetric cone, whose shape is highly symmetric.
 In particular, $N_I = 1$ holds,
 i.e., $\St_I = \Scalar \cong \mS_+(\Complex) \cong \Real_+$,
 where $\Real_+$ is the set of all nonnegative real numbers.
 $\Proc_{A \to B}$ is isomorphic to the space of all CP
 maps from $\mS(\Complex^\NA)$ to $\mS(\Complex^\NB)$.
 In the examples of quantum theory, we will identify a process with its corresponding CP map.
 Also, we will identify a state (or effect) with the corresponding positive semidefinite matrix.
\end{ex}

\subsection{Sequential and parallel compositions}

Two processes can be composed sequentially whenever the output system of one and
the input system of the other are the same.
The sequential composition of $f \in \Proc_{A \to B}$ and $g \in \Proc_{B \to C}$
is also a process, denoted as $g \circ f \in \Proc_{A \to C}$.
When we write $g \circ f$, we always assume that the output system of $f$ and
the input system of $g$ are equal.
For any $\cket{\rho} \in \St_A$ and $\cra{e} \in \Eff_A$,
$\cra{e} \circ \cket{\rho} \in \Scalar$ is denoted by $\craket{e|\rho}$.
$g \circ f$ and $\craket{e|\rho}$ are respectively depicted as
\begin{alignat}{1}
 \InsertPDF{figures/process_fg_circ.pdf} ~\raisebox{.0em}{.}
 \label{eq:process_circ}
\end{alignat}

Any two systems and processes can be composed in parallel.
The parallel composition of two systems, $A$ and $B$, is a system, denoted by $A \otimes B$.
Assume that $I \otimes A = A = A \otimes I$ holds.
The parallel composition of $f \in \Proc_{A \to B}$ and $g \in \Proc_{C \to D}$
is a process from $A \otimes C$ to $B \otimes D$, denoted as $f \otimes g$.
$f \otimes g$ is diagrammatically depicted as
\begin{alignat}{1}
 \InsertPDF{figures/process_fg_otimes.pdf} ~\raisebox{.5em}{.}
 \label{eq:process_otimes}
\end{alignat}

A collection of connected processes will be called a \termdef{diagram}.
These sequential and parallel compositions are associative, e.g.,
$h \circ (g \circ f) = (h \circ g) \circ f$ holds for
any $f \in \Proc_{A \to B}$, $g \in \Proc_{B \to C}$, and $h \in \Proc_{C \to D}$.
Assume that
\begin{alignat}{1}
 (g_1 \otimes g_2) \circ (f_1 \otimes f_2) &= (g_1 \circ f_1) \otimes (g_2 \circ f_2),
 \label{eq:process_product2}
\end{alignat}
or diagrammatically
\begin{alignat}{1}
 \InsertPDF{figures/process_product2.pdf} ~\raisebox{.5em}{,}
 \label{eq:process_product2_diagram}
\end{alignat}
holds for four any processes $f_1$, $f_2$, $g_1$, and $g_2$,
where the auxiliary lines (dashed lines) are drawn to guide the eye.
For any scalar $a$ and process $f$, $a \otimes f$ is denoted by $af$ or $a \cdot f$.
One can see that $(ag) \circ f = a(g \circ f) = g \circ (af)$,
$(af) \otimes h = a(f \otimes h) = f \otimes (ah)$, and
$a \circ b = ab$ hold for any scalars $a$ and $b$ and any processes $f$, $g$, and $h$.

\begin{ex}[quantum theory]
 $N_{A \otimes B} = \NA\NB$ holds for any systems $A$ and $B$.
 Since $\St_{I \otimes A} \cong \St_{A \otimes I} \cong \St_A \cong \mS_+(\Complex^\NA)$
 holds from $N_I = 1$,
 $I \otimes A$ and $A \otimes I$ can be identified with $A$.
 For two processes $f$ and $g$, $g \circ f$ is the CP map
 satisfying $(g \circ f)[\cket{\rho}] \coloneqq g[f[\cket{\rho}]]$.
 In particular, $\craket{e|\rho} = \Tr[\cra{e} \cdot \cket{\rho}]$ holds
 ($\cdot$ is the matrix product).
 For two states $\cket{\rho}$ and $\cket{\sigma}$,
 $\cket{\rho} \otimes \cket{\sigma}$ is
 the tensor product of the matrices $\cket{\rho}$ and $\cket{\sigma}$.
 $f \otimes g$ is the CP map defined as
 $(f \otimes g)[\cket{\rho} \otimes \cket{\sigma}] \coloneqq
 f[\cket{\rho}] \otimes g[\cket{\sigma}]$.
\end{ex}

\subsection{Identity processes and discarding effects}

An \termdef{identity process} on $A$, denoted by $\id_A$ or simply $\id$,
is the process satisfying
(a) $f \circ \id_A = f = \id_B \circ f$,
(b) $\id_A \otimes \id_B = \id_{A \otimes B}$,
and (c) $f \otimes \id_I = f = \id_I \otimes f$,
where $A$ and $B$ are any systems and $f$ is any process from $A$ to $B$.
Assume that there exists $\id_A$ for each system $A$.
Diagrammatically, $\id_A$ is depicted as
\begin{alignat}{1}
 \InsertPDF{figures/process_id_def.pdf} ~\raisebox{.0em}{.}
 \label{eq:process_id_def}
\end{alignat}
$\id_I$ is depicted as empty space.
The above property (a) is depicted as:
\begin{alignat}{1}
 \InsertPDF{figures/process_id_f.pdf} ~\raisebox{.0em}{,}
 \label{eq:process_id_f}
\end{alignat}
where the auxiliary boxes indicate the identity processes.
This, intuitively, implies that the length of lines does not change diagrams.
It also follows from Eq.~\eqref{eq:process_id_f} that, for any processes $f$ and $g$,
\begin{alignat}{1}
 \InsertPDF{figures/process_fg.pdf}
 \label{eq:process_fg}
\end{alignat}
holds.
Intuitively, this yields that the vertical shifts of processes do not affect diagrams.

Assume that, for any systems $A$ and $B$,
there exists a process $\cross_{A \otimes B} \in \Proc_{A \otimes B \to B \otimes A}$,
called a \termdef{swap process} and diagrammatically depicted by
\begin{alignat}{1}
 \InsertPDF{figures/process_cross_def.pdf}
 ~\raisebox{.0em}{,}
 \label{eq:cross_def}
\end{alignat}
such that
\begin{alignat}{1}
 \InsertPDF{figures/process_cross_f.pdf} \label{eq:cross_f}
\end{alignat}
(i.e., $\cross_{A',B} \circ (f \otimes \id_B) \circ \cross_{B,A} = \id_B \otimes f$)
holds for any systems $A$, $A'$, and $B$ and any process $f \in \Proc_{A \to A'}$.
Also, assume that $\cross_{A,I} = \id_A = \cross_{I,A}$,
i.e.,
\begin{alignat}{1}
 \InsertPDF{figures/process_id_curve.pdf} ~\raisebox{.0em}{,}
 \label{eq:process_id_def}
\end{alignat}
and $\cross_{A,B \otimes C} = (\id_B \otimes \cross_{A,C}) \circ (\cross_{A,B} \otimes \id_C)$
hold for any systems $A$, $B$, and $C$.

Let us consider a trivial measurement of system $A$, which always gives the same outcome
regardless of the input state.
The effect representing an event associated to the outcome of a trivial measurement
is called a \termdef{discarding effect} and denoted by $\cra{\gdis_A}$,
or simply $\cra{\gdis}$, which is depicted as
\begin{alignat}{1}
 \InsertPDF{figures/discard.pdf} ~\raisebox{.0em}{.}
 \label{eq:discard}
\end{alignat}
$\cra{\gdis_A}$ has a natural operational intuition:
one performs any measurement on system $A$ and then discards the results.

\begin{ex}[quantum theory]
 The identity process is the identity map.
 The discarding effect of $A$ is the identity matrix of order $\NA$,
 denoted by $\ident_\NA$.
\end{ex}

\subsection{Probabilistic behavior}

Assume that $\Scalar = \Real_+$ holds, i.e., each scalar is identified with a nonnegative real number.
Let $\ScalarF \coloneqq \{ p \in \Scalar : p \le 1 \}$.
A state $\cket{\rho} \in \St_A$ is called \termdef{feasible} if $\craket{\gdis_A|\rho} \in \ScalarF$
holds; in particular, $\cket{\rho}$ is called \termdef{normalized} (or deterministic)
if $\craket{\gdis_A|\rho} = 1$ holds.
Let $\StF_A$ and $\StN_A$ be, respectively, the sets of all feasible and normalized states of $A$.
A process $f \in \Proc_{A \to B}$ is called \termdef{feasible} if
$(f \otimes \id_E) \circ \cket{\sigma} \in \StF_{B \otimes E}$ holds for any system $E$
and $\cket{\sigma} \in \StF_{A \otimes E}$.
Also, $f \in \Proc_{A \to B}$ is called \termdef{deterministic} if
$(f \otimes \id_E) \circ \cket{\sigma} \in \StN_{B \otimes E}$ holds for any system $E$
and $\cket{\sigma} \in \StN_{A \otimes E}$.
Let $\ProcF_{A \to B}$ and $\ProcD_{A \to B}$ be, respectively,
the sets of all feasible and deterministic processes from $A$ to $B$.
Also, let $\EffF_A \coloneqq \ProcF_{A \to I}$.
Assume that each scalar consisting of the sequential and/or parallel compositions of $k$
feasible processes $f_1,\ldots,f_k$, for example, the scalar depicted by
\begin{alignat}{1}
 \InsertPDF{figures/process_f1_fk.pdf}
 \label{eq:process_f1_fk} 
\end{alignat}
is the probability of the joint occurrence of $f_1,\ldots,f_5$%
\footnote{In OPTs, a certain set of scalars is associated with a probability distribution.
See, e.g., Ref.~\cite{Chi-Dar-Per-2010} for details.}.
In particular, for any $p,q \in \ScalarF$, $pq$, i.e.,
the probability of the joint occurrence of $p$ and $q$,
is the product of real numbers $p$ and $q$.
Assume $\cra{\gdis_I} = 1$ and $\cra{\gdis_A} \in \ProcD_{A \to I}$.
1 is the unique deterministic scalar and $\id_I = 1$ holds.
It follows that any process consisting of the sequential and/or parallel compositions of
deterministic (resp. feasible) processes is deterministic (resp. feasible).
One can easily see $\StF_A = \ProcF_{I \to A}$, $\ScalarF = \StF_I = \ProcF_{I \to I}$,
$\StN_A = \ProcD_{I \to A}$, and $\ProcD_{A \to B} \subset \ProcF_{A \to B}$.
Any $\cket{\rho} \in \StF_A$ is in the form $\cket{\rho} = p \cket{\rho^\N}$
with $p \in \ScalarF$ and $\cket{\rho^\N} \in \StN_A$
(note that $p = \craket{\gdis|\rho}$ holds from $\craket{\gdis|\rho^\N} = 1$),
which means that $\cket{\rho}$ can be identified with
the process preparing the normalized state $\cket{\rho^\N}$ with probability $p$.

From the definition, a scalar larger than 1 is unfeasible.
Unfeasible scalars cannot be interpreted as probabilities
and thus are not intuitive.
However, it is mathematically convenient to consider unfeasible scalars,
so we assume $\Scalar = \Real_+$.
Similarly, $\Proc_{A \to B}$ is defined as
\begin{alignat}{1}
 \Proc_{A \to B} &\coloneqq \{ a f : a \in \Scalar, f \in \ProcF_{A \to B} \}.
 \label{eq:process_Proc_def}
\end{alignat}
Although unfeasible processes exist in each process space $\Proc_{A \to B}$
(i.e., $\ProcF_{A \to B} \subsetneq \Proc_{A \to B}$),
Eq.~\eqref{eq:process_Proc_def} implies that any unfeasible process is expressed as
scalar multiplication of a feasible process.

We consider the following diagram, denoted by $u:\Proc_{A \to B} \to \Scalar$,
that maps a process $f \in \Proc_{A \to B}$ to
a scalar $u(f) \coloneqq \cra{u_2} \circ (f \otimes \id_E) \circ \cket{u_1} \in \Scalar$,
where $\cket{u_1} \in \St_{A \otimes E}$ and $\cra{u_2} \in \Eff_{B \otimes E}$ hold.
$u$ can be interpreted as a set of a system $E$, a state $\cket{u_1} \in \St_{A \otimes E}$,
and an effect $\cra{u_2} \in \Eff_{B \otimes E}$,
which is diagrammatically depicted as
\begin{alignat}{1}
 \InsertPDF{figures/u_scalar.pdf} ~\raisebox{1.0em}{.}
 \label{eq:u_scalar}
\end{alignat}
Any scalar that includes $f \in \Proc_{A \to B}$ is expressed in the form $u(f)$
with some diagram $u:\Proc_{A \to B} \to \Scalar$;
for example,
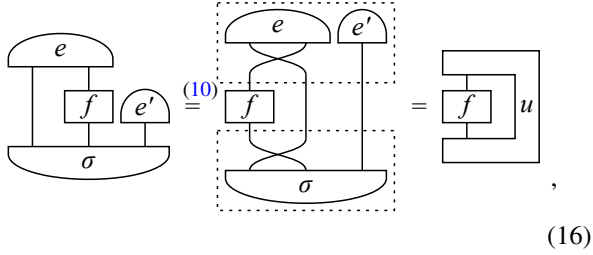
\begin{alignat}{1}
 \begin{overpic}{figures/uf_example.pdf}
  \put(33,22){\footnotesize\eqref{eq:cross_f}}
 \end{overpic}
 ~\raisebox{1.0em}{,}
 \label{eq:uf_example}
\end{alignat}
where $\cket{u_1}$ and $\cra{u_2}$ are, respectively,
the state and effect enclosed by the auxiliary boxes.

For two processes $f,f' \in \Proc_{A \to B}$, $f = f'$ is defined as
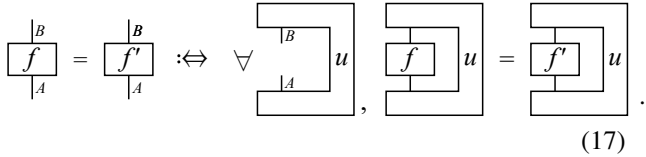
\begin{alignat}{1}
 \InsertPDF{figures/process_eq.pdf} ~\raisebox{.5em}{.}
 \label{eq:process_eq}
\end{alignat}
This means that $f = f'$ holds if they are indistinguishable in a probabilistic sense.
For $f, f' \in \Proc_{A \to B}$, $f \eqlocal f'$ is defined as
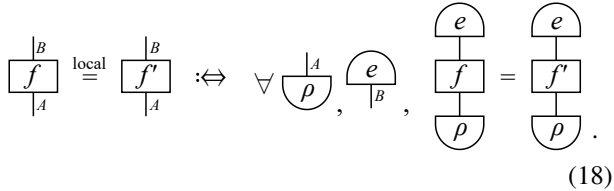
\begin{alignat}{1}
 \InsertPDF{figures/process_eqlocal.pdf} ~\raisebox{.5em}{.}
 \label{eq:process_eqlocal}
\end{alignat}
One can easily see that $f \eqlocal f'$ holds if $f = f'$ holds,
but the converse is not necessarily true.
It follows that, in the case of $A = I$ or $B = I$, $f \eqlocal f'$ and $f = f'$ are
always the same, which means
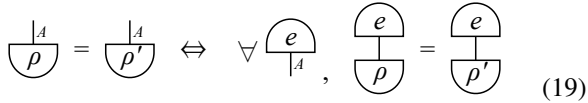
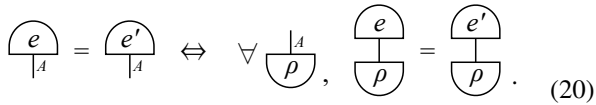
\begin{alignat}{1}
 \InsertPDF{figures/process_eq_state.pdf}
 \label{eq:process_eq_state}
\end{alignat}
and
\begin{alignat}{1}
 \InsertPDF{figures/process_eq_eff.pdf} ~\raisebox{.5em}{.}
 \label{eq:process_eq_eff}
\end{alignat}
$\cra{\gdis_A}$ is the unique deterministic effect of $A$ and that
$\cra{\gdis_{A \otimes B}} \coloneqq \cra{\gdis_A} \otimes \cra{\gdis_B}$ holds,
which is depicted as
\begin{alignat}{1}
 \InsertPDF{figures/discard_AB.pdf} ~\raisebox{.5em}{.}
 \label{eq:discard_AB}
\end{alignat}

\begin{ex}[quantum theory]
 Since $\craket{\gdis_A|\rho} = \Tr~\cket{\rho}$ holds,
 $\cket{\rho} \in \StF_A$ means $\Tr~\cket{\rho} \le 1$.
 Also, $\cket{\rho} \in \StN_A$ means $\Tr~\cket{\rho} = 1$.
 $f \in \ProcF_{A \to B}$ holds if and only if $f$ is a trace non-increasing CP map.
 Also, $f \in \ProcD_{A \to B}$ means that $f$ is a trace-preserving (TP) CP map.
 $\cra{e} \in \EffF_{A \to B}$ means that the maximal eigenvalue of the matrix $\cra{e}$
 is not larger than 1.
 For any $f,f' \in \Proc_{A \to B}$, $f = f'$ and $f \eqlocal f'$ are equivalent.
\end{ex}

\subsection{Process space spans vector space}

Assume that, for any two feasible processes $g_1,g_2 \in \ProcF_{A \to B}$
and any $p \in \ScalarF$,
there exists a feasible process $h \in \ProcF_{A \to B}$ satisfying
\begin{alignat}{1}
 \InsertPDF{figures/process_weighted_sum.pdf}
 \label{eq:process_weighted_sum}
\end{alignat}
for any $u:\Proc_{A \to B} \to \Scalar$.
Such a process $h$ is denoted by $p g_1 + (1-p) g_2$.
This can be interpreted as a probabilistic mixture of
$g_1$ and $g_2$ with probabilities $p$ and $1-p$.
For any two processes $f_1,f_2 \in \Proc_{A \to B}$,
the process $h'$ satisfying
\begin{alignat}{2}
 \InsertPDF{figures/process_sum2.pdf}
 \label{eq:process_sum2}
\end{alignat}
for any $u:\Proc_{A \to B} \to \Scalar$
is called the \termdef{sum of $f_1$ and $f_2$} and denoted by $f_1 + f_2$.
One can easily see that $f_1 + f_2 \in \Proc_{A \to B}$ holds
for any $f_1,f_2 \in \Proc_{A \to B}$.

Let $\mI_n \coloneqq \{ 1,\ldots,n \}$.
A set of effects $\{ \cra{e_m} \in \Eff_A \}_{m \in \mI_M}$ is called
a \termdef{measurement} if
\begin{alignat}{1}
 \InsertPDF{figures/meas_sum_normal.pdf}
 \label{eq:meas_sum_normal}
\end{alignat}
holds.
This means that the sum of probabilities over all possible outcomes is 1
whenever one performs a measurement on a normalized state.
Equation~\eqref{eq:meas_sum_normal} is equivalent to
\begin{alignat}{1}
 \InsertPDF{figures/meas_sum.pdf} ~\raisebox{.3em}{,}
 \label{eq:meas_sum}
\end{alignat}
and thus, from Eq.~\eqref{eq:process_eq_eff}, it follows that
$\{ \cra{e_m} \in \Eff_A \}_{m \in \mI_M}$ is a measurement if and only if
\begin{alignat}{1}
 \InsertPDF{figures/meas_sum2.pdf}
 \label{eq:meas_sum2}
\end{alignat}
holds.

We can consider the real vector space $\Vec_{A \to B}$ spanned by $\Proc_{A \to B}$,
whose elements are formal sums of the form $\sum_i a_i f_i$ with $a_i \in \Real$
and $f_i \in \Proc_{A \to B}$,
where the element $\ol{f} \coloneqq \sum_i a_i f_i$ satisfies
\begin{alignat}{1}
 \InsertPDF{figures/fsum_u_ex.pdf}
 \label{eq:fsum_u_ex}
\end{alignat}
for any $u:\Proc_{A \to B} \to \Scalar$.
We will call an element of $\Vec_{A \to B}$ an \termdef{extended process},
which is denoted with an overline such as $\ol{f}$
(unless it is clearly a process).
Equation~\eqref{eq:fsum_u_ex} implies that the diagram $u$ distributes over addition.
$\Proc_{A \to B} \subseteq \Vec_{A \to B}$ obviously holds.
Any $\ol{f} \in \Vec_{A \to B}$ is expressed by
$\ol{f} = f_+ - f_-$ with some $f_+, f_- \in \Proc_{A \to B}$.
As well as processes, extended processes can be composed sequentially and in parallel.
Specifically, for any $\ol{f} \coloneqq \sum_i a_i f_i \in \Vec_{A \to B}$,
$\ol{g} \coloneqq \sum_j b_j g_j \in \Vec_{B \to C}$,
and $\ol{h} \coloneqq \sum_k c_k h_k \in \Vec_{C \to D}$
with $a_i, b_j, c_k \in \Real$, $f_i \in \Proc_{A \to B}$, $g_j \in \Proc_{B \to C}$,
and $h_k \in \Proc_{C \to D}$,
$\ol{g} \circ \ol{f} = \sum_i \sum_j a_i b_j (g_j \circ f_i)$ and
$\ol{f} \otimes \ol{h} = \sum_i \sum_k a_i c_k (f_i \otimes h_k)$ hold.
This is diagrammatically depicted as
\begin{alignat}{1}
 \InsertPDF{figures/process_ex_sequential.pdf}
 \label{eq:process_ex_sequential}
\end{alignat}
and
\begin{alignat}{1}
 \InsertPDF{figures/process_ex_parallel.pdf} ~\raisebox{.5em}{.}
 \label{eq:process_ex_parallel}
\end{alignat}
Let $\Vec_A \coloneqq \Vec_{I \to A}$ and $\Vec_A^* \coloneqq \Vec_{A \to I}$;
then, $\Vec_A^*$ can be regarded as the dual vector space of $\Vec_A$.
We can easily verify that, for any systems $A$ and $B$, $\Proc_{A \to B}$ is a convex cone.
In particular, $\St_A$ is a convex cone in $\Vec_A$.
The dimension of the real vector space $\Vec_A$ is called
the \termdef{dimension} of $A$.
Assume that $\Proc_{A \to B}$ is closed.
$f \in \Proc_{A \to B}$ is called \termdef{atomic} if $f_1 \propto f$ holds
for any $f_1,f_2 \in \Proc_{A \to B}$ satisfying $f = f_1 + f_2$,
where we denote $f_1 \propto f$ if there exists a scalar $a$ satisfying
either $f_1 = a f$ or $a f_1 = f$.
In particular, states and effects are also called \termdef{pure} if they are atomic.
Let $\StP_A$ be the set of all pure states of $A$.
Also, let $\StNP_A$ be the set of all normalized pure states,
i.e., $\StNP_A \coloneqq \StN_A \cap \StP_A$.
Any $\cket{\rho} \in \St_A$ can be expressed in the form
\begin{alignat}{1}
 \InsertPDF{figures/process_state_decomposed.pdf}
 ~\raisebox{.0em}{,}
 \label{eq:process_state_decomposed}
\end{alignat}
where $\cket{\psi_1},\ldots,\cket{\psi_k} \in \StNP_A$ and
$p_1,\ldots,p_k \in \Real_+$.

One can show that, for any $f \in \Proc_{A \to B}$, we have
\newsavebox{\boxprocessdet}
\sbox{\boxprocessdet}{\InsertPDF{figures/process_deterministic.pdf}}
\newlength{\bwprocessdet}
\settowidth{\bwprocessdet}{\usebox{\boxprocessdet}}
\begin{alignat}{1}
 f \in \ProcD_{A \to B} &\quad\Leftrightarrow\quad
 \parbox{\bwprocessdet}{\usebox{\boxprocessdet}}
 ~\raisebox{.0em}{.}
 \label{eq:process_deterministic}
\end{alignat}
Similarly to Eq.~\eqref{eq:process_deterministic},
an extended process $\ol{f} \in \Vec_{A \to B}$ is called \termdef{deterministic}
if $\cra{\gdis} \circ \ol{f} = \cra{\gdis}$ holds.

An extended process $\ol{f} \in \Vec_{A \to B}$ is called \termdef{reversible} if
there exists $\ol{g} \in \Vec_{B \to A}$, called an \termdef{inverse} of $\ol{f}$, such that
$\ol{g} \circ \ol{f} \eqlocal \id_A$ and $\ol{f} \circ \ol{g} \eqlocal \id_B$.
Such $\ol{g}$ is also reversible.
We denote $A \cong B$ if there exists a reversible extended process from $A$ to $B$.

A process $f \in \Proc_{A \otimes B \to C \otimes D}$ is called \termdef{separable} if
it can be expressed in the form
\begin{alignat}{1}
 \InsertPDF{figures/process_separable.pdf}
 \label{eq:process_separable}
\end{alignat}
with $g_i \in \Proc_{A \to C}$ and $h_i \in \Proc_{B \to D}$.

\begin{ex}[quantum theory]
 The sum of processes is equal to the sum of CP maps.
 In particular, the sum of states (or effects) is the sum of matrices.
 $\Pi \coloneqq \{ \cra{e_m} \in \Eff_A \}_{m \in \mI_M}$ is a measurement if and only if
 $\sum_{m=1}^M \cra{e_m} = \ident_\NA$ holds, i.e.,
 $\Pi$ is a positive operator-valued measure (POVM)
 (note that each effect $\cra{e_m}$ is a positive semidefinite matrix).
 In fully quantum theory, $\Vec_{A \to B}$ is isomorphic to the space of all
 linear maps from $\mS(\Complex^\NA)$ to $\mS(\Complex^\NB)$
 (which are also called Hermitian-preserving maps);
 in particular, $\Vec_A$ and $\Vec_A^*$ are isomorphic to
 $\mS(\Complex^\NA)$.
 $\cket{\psi} \in \St_A$ is pure if and only if
 $\cket{\psi} = \ket{\psi}\bra{\psi}$ holds for some vector $\ket{\psi}$.
\end{ex}

\section{OPTs with classical systems}

\subsection{Classical systems}

We will call an $M$-dimensional system $C$ \termdef{classical} if
there exist $M$ normalized pure states $\cket{1},\ldots,\cket{M} \in \StNP_C$ and
a measurement $\{ \cra{m} \in \Eff_C \}_{m \in \mI_M}$ satisfying
\begin{alignat}{1}
 \InsertPDF{figures/classical_st_meas.pdf}
 \label{eq:classical_st_meas}
\end{alignat}
and
\begin{alignat}{1}
 \InsertPDF{figures/classical_id.pdf} ~\raisebox{.5em}{,}
 \label{eq:classical_id}
\end{alignat}
where $\delta_{m,n}$ is the Kronecker delta.
A classical system is depicted as the dotted line.
One can easily see that any $\cket{\rho} \in \St_C$ is expressed in the form
\begin{alignat}{1}
 \begin{overpic}{figures/classical_state_decomposed.pdf}
  \put(18,26){\footnotesize\eqref{eq:classical_id}}
 \end{overpic}
 ~\raisebox{2.0em}{,}
 \label{eq:classical_state_decomposed}
\end{alignat}
where $p_m \coloneqq \craket{m|\rho} \in \Scalar$.
This immediately gives that $\StNP_C = \{ \cket{m} \}_{m \in \mI_M}$ holds.
Indeed, from Eq.~\eqref{eq:classical_state_decomposed},
$\cket{\rho} \in \StNP_C$ holds if and only if $p_m = \delta_{i,m}$ holds
for some $i \in \mI_M$.
In what follows, we consider an OPT that has an $M$-dimensional classical system $C$.
Note that a classical system is not intrinsically necessary for investigating
a state discrimination problem in an OPT,
but it helps us to express this problem in straightforward diagrammatic terms.

It follows from Eq.~\eqref{eq:classical_id} that
any state of $C \otimes A$ (where $A$ is an arbitrary system that is not classical in general)
is separable.
Indeed, one can easily obtain
\begin{alignat}{1}
 \begin{overpic}{figures/classical_state_composite_proof.pdf}
  \put(16,24){\footnotesize\eqref{eq:classical_id}}
 \end{overpic}
 ~\raisebox{2.0em}{,}
 \label{eq:classical_state_composite_proof}
\end{alignat}
where the state enclosed by the auxiliary box is denoted by $\cket{\rho'_m}$.
Similarly, any state of $A \otimes C$, effect of $C \otimes A$, and effect of $A \otimes C$
are separable.
Moreover, for any $f \in \Proc_{C \to A}$ with a classical system $C$,
$\cket{\sigma} \in \St_{C \otimes E}$, and $\cra{e} \in \Eff_{A \otimes E}$,
Eq.~\eqref{eq:classical_state_composite_proof} yields
\begin{alignat}{1}
 \InsertPDF{figures/classical_eq.pdf} ~\raisebox{.5em}{,}
 \label{eq:classical_eq}
\end{alignat}
where the effect enclosed by the auxiliary box is denoted by $\cra{e_m}$.
One can see from Eq.~\eqref{eq:classical_eq} that,
for any $f,f' \in \Proc_{C \to A}$,
\begin{alignat}{1}
 f = f' &\quad\Leftrightarrow\quad f \eqlocal f'
 \label{eq:fca_eq}
\end{alignat}
holds.
Similarly, Eq.~\eqref{eq:fca_eq} holds for any $f,f' \in \Proc_{A \to C}$
with a classical system $C$.

\subsection{$\cket{\cup}$ and $\cra{\cap}$}

We will introduce the state $\cket{\cup}$ and the effect $\cra{\cap}$
defined as
\begin{alignat}{1}
 \InsertPDF{figures/classical_cup.pdf} ~\raisebox{.5em}{.}
 \nonumber \\
 \label{eq:classical_cup}
\end{alignat}
Clearly, we have
\begin{alignat}{1}
 \begin{overpic}{figures/classical_cup_st_e.pdf}
  \put(18,13){\footnotesize\eqref{eq:classical_cup}}
  \put(57,13){\footnotesize\eqref{eq:classical_st_meas}}
  \put(75,17){\footnotesize\eqref{eq:classical_cup}}
  \put(76,13){\footnotesize\eqref{eq:classical_st_meas}}
 \end{overpic}
 ~\raisebox{.5em}{,}
 \label{eq:classical_cup_st_e}
\end{alignat}
\begin{alignat}{1}
 \begin{overpic}{figures/classical_cup_st_e2.pdf}
  \put(18,12){\footnotesize\eqref{eq:classical_cup}}
  \put(57,12){\footnotesize\eqref{eq:classical_st_meas}}
  \put(75,16){\footnotesize\eqref{eq:classical_cup}}
  \put(76,12){\footnotesize\eqref{eq:classical_st_meas}}
 \end{overpic}
 ~\raisebox{.5em}{,}
 \label{eq:classical_cup_st_e2}
\end{alignat}
and
\begin{alignat}{1}
 \begin{overpic}{figures/classical_cup_st.pdf}
  \put(39,21){\footnotesize\eqref{eq:classical_cup_st_e}}
  \put(75,21){\footnotesize\eqref{eq:classical_st_meas}}
 \end{overpic}
 ~\raisebox{.5em}{.}
 \label{eq:classical_cup_st}
\end{alignat}
One can also show
\begin{alignat}{1}
 \InsertPDF{figures/classical_cup_cap_id.pdf} ~\raisebox{.5em}{.}
 \label{eq:classical_cup_cap_id}
\end{alignat}
Intuitively, this means that a curved line consisting of `$\cup$' and `$\cap$' can be yanked.
The first equality of Eq.~\eqref{eq:classical_cup_cap_id} is obtained from
\begin{alignat}{1}
 \begin{overpic}{figures/classical_cup_cap_id_proof.pdf}
  \put(23,14){\footnotesize\eqref{eq:classical_cup}}
  \put(68,14){\footnotesize\eqref{eq:classical_st_meas}}
  \put(92,14){\footnotesize\eqref{eq:classical_id}}
 \end{overpic}
 ~\raisebox{.5em}{.}
 \label{eq:classical_cup_cap_id_proof}
\end{alignat}
The same is true for the second equality of Eq.~\eqref{eq:classical_cup_cap_id}.
$\cket{\gdisT} \in \St_C$ is defined as
\begin{alignat}{1}
 \InsertPDF{figures/classical_chi.pdf}~\raisebox{.5em}{;}
 \label{eq:classical_chi}
\end{alignat}
then, we have
\begin{alignat}{1}
 \begin{overpic}{figures/classical_chi_gdis.pdf}
  \put(53,40){\footnotesize\eqref{eq:meas_sum2}}
  \put(50,29){\footnotesize\eqref{eq:classical_cup_st_e2}}
 \end{overpic}
 ~\raisebox{.5em}{.}
 \label{eq:classical_chi_gdis}
\end{alignat}

\subsection{Expression of measurements}

In an OPT with a classical system $C$,
any measurement $\{ \cra{e_m} \}_{m \in \mI_M}$ can be expressed by
the process $e \coloneqq \sum_{n=1}^M \cket{n} \cra{e_n} \in \Proc_{A \to C}$.
Diagrammatically,
\begin{alignat}{1}
 \InsertPDF{figures/meas.pdf} ~\raisebox{.5em}{.}
 \label{eq:meas}
\end{alignat}
Each effect $\cra{e_n}$ with $n \in \mI_M$ is obtained from
\begin{alignat}{1}
 \begin{overpic}{figures/meas2.pdf}
  \put(41,48){\footnotesize\eqref{eq:classical_st_meas}}
 \end{overpic}
 ~\raisebox{.5em}{.}
 \label{eq:meas2}
\end{alignat}
One can easily verify that $\cra{\gdis_C} \circ e = \cra{\gdis_A}$,
i.e., $e$ is deterministic from Eq.~\eqref{eq:process_deterministic}.
Conversely, any deterministic process from $A$ to $C$ can be interpreted
as a measurement since $e' \in \ProcD_{A \to C}$ can be depicted as
\begin{alignat}{1}
 \begin{overpic}{figures/meas_determine.pdf}
  \put(18,25){\footnotesize\eqref{eq:classical_id}}
 \end{overpic}
 ~\raisebox{.5em}{,}
 \label{eq:meas_determine}
\end{alignat}
where $\cra{e'_n} \coloneqq \cra{n} \circ e'$.
For each system $A$, $\ProcD_{A \to C}$ is denoted by $\Meas_A^C$ or simply $\Meas_A$.

\subsection{State preparation}

We assume that one of $M$ normalized states of system $A$,
$\cket{\rho^\N_1},\ldots,\cket{\rho^\N_M} \in \StN_A$, is prepared
with prior probabilities $\xi_1,\ldots,\xi_M$ with $\sum_{m=1}^M \xi_m = 1$.
Let $\cket{\rho_m} \coloneqq \xi_m \cket{\rho^\N_m}$.
We consider the process $\rho \in \Proc_{C \to A}$ defined as
\begin{alignat}{1}
 \InsertPDF{figures/st.pdf}~\raisebox{.5em}{.}
 \label{eq:st}
\end{alignat}
Each $\cket{\rho_n}$ is obtained from
\begin{alignat}{1}
 \begin{overpic}{figures/st2.pdf}
  \put(41,45){\footnotesize\eqref{eq:classical_st_meas}}
 \end{overpic}
 ~\raisebox{.5em}{.}
 \label{eq:st2}
\end{alignat}
Note that $\rho$ is not deterministic unless $M = 1$.
We have
\begin{alignat}{1}
 \begin{overpic}{figures/st_chi.pdf}
  \put(33,48){\footnotesize\eqref{eq:st}}
  \put(33,36){\footnotesize\eqref{eq:classical_chi}}
 \end{overpic}
 ~\raisebox{.5em}{.}
 \label{eq:st_chi}
\end{alignat}
$\sum_{m=1}^M \cket{\rho_m}$ is obviously normalized.
We will call a process $\rho \in \Proc_{C \to A}$ with
$\rho \circ \cket{\gdisT} \in \StN_C$ a \termdef{state preparation}.
Let $\Prep_A$ be the set of all state preparations.

\section{Minimum-error measurement for group covariant states} \label{sec:opt}

\subsection{Discrimination problem} \label{subsec:opt_discriminate}

Let us review the problem of discriminating a given set of $M$ known normalized states
with given prior probabilities
\cite{Kim-Miy-Ima-2009,Nui-Kim-Miy-2010,Bae-2012}.
Here, we consider the following scenario:
One party (Charlie) randomly chooses one of the $M$ states
$\cket{\rho_1^\N},\ldots,\cket{\rho_M^\N} \in \StN_A$
with prior probabilities $\xi_1,\ldots,\xi_M$.
Such a process is expressed by the state preparation $\rho \in \Prep_A$ depicted by Eq.~\eqref{eq:st}.
Since she knows which state she has, we can interpret that she has the following state
\begin{alignat}{1}
 \begin{overpic}{figures/st_charlie.pdf}
  \put(56,26){\footnotesize\eqref{eq:classical_cup}}
 \end{overpic}
 ~\raisebox{.5em}{.}
 \label{eq:st_charlie}
\end{alignat}
Indeed, by performing a measurement $\{ \cra{m} \}_{m \in \mI_M}$ on the classical system $C$,
she can always determine which state she has.
Charlie sends the state to the other party (Alice).
Alice knows the possible states $\cket{\rho_1^\N},\ldots,\cket{\rho_M^\N}$ and
their prior probabilities but does not know which state Charlie sent.
We can interpret that Alice gets the following normalized state
\begin{alignat}{1}
 \begin{overpic}{figures/st_alice.pdf}
  \put(57,35){\footnotesize\eqref{eq:classical_chi_gdis}}
 \end{overpic}
 ~\raisebox{.5em}{.}
 \label{eq:st_alice}
\end{alignat}
What Alice has to do is to perform a measurement that will correctly discriminate
between the states $\cket{\rho_1^\N},\ldots,\cket{\rho_M^\N}$ with high probability.
Alice performs a measurement 
$e \coloneqq \sum_{m=1}^M \cket{m} \circ \cra{e_m} \in \Meas_A$ with $\cra{e_m} \in \Eff_A$
to discriminate between the states as accurately as possible.
After that, Charlie and Alice check whether Alice correctly determines the state.
In this paper, we apply the strategy that maximizes the average success probability.
This probability is given as a function of the measurement $e$, denoted by $\PS(e)$,
which is expressed by
\begin{alignat}{1}
 \InsertPDF{figures/PS.pdf} ~\raisebox{.5em}{.}
 \label{eq:PS}
\end{alignat}
Note that $\PS(e)$ is also expressed by
\begin{alignat}{2}
 \PS(e) &= \sum_{m=1}^M \craket{e_m|\rho_m} &&=
 \sum_{m=1}^M \cra{m} \circ e \circ \rho \circ \cket{m}.
\end{alignat}
A measurement that maximizes the average success probability is called
a \termdef{minimum-error measurement}.
The problem of finding a minimum-error measurement is formulated as
\begin{alignat}{1}
 \begin{array}{ll}
  \mbox{maximize} & \PS(e) \\
  \mbox{subject~to} & e \in \Meas_A \\
 \end{array} \label{eq:main}
\end{alignat}
with variable $e$.
If we want to optimize over a restricted class of measurements,
denoted by $\Measgen_A$ with $\Measgen_A \subseteq \Meas_A$,
then we consider the following problem:
\begin{alignat}{1}
 \begin{array}{ll}
  \mbox{maximize} & \PS(e) \\
  \mbox{subject~to} & e \in \Measgen_A. \\
 \end{array} \label{eq:main_gen}
\end{alignat}

\subsection{Group action}

We use group theory to represent the symmetric properties of a state preparation.
Let $\mG$ be a group and $1 \in \mG$ be its identity element.
Let $\Aut(\mZ)$ be the group of all automorphisms of an object $\mZ$.
A map $\psi: \mG \ni g \mapsto \psi_g \in \Aut(\mZ)$ is called a (right) \termdef{group action}
of $\mG$ on $\mZ$
if $\psi_{gh}(z) = \psi_h[\psi_g(z)]$ $~(\forall g,h \in \mG)$ and $\psi_1(z) = z$ hold
for any $z \in \mZ$.
In the following, we give two examples of group actions.

The first example is a group action, $\tau:\mG \to \Aut(\mI_M)$, of $\mG$ on $\mI_M$.
$\Aut(\mI_M)$ is the group of all permutations of $\mI_M$.
We will identify $\mI_M$ with $\{ \cra{m} \in \Eff_C \}_{m \in \mI_M}$;
then, we have
\begin{alignat}{1}
 \InsertPDF{figures/sym_pi_id.pdf} ~\raisebox{.5em}{.}
 \label{eq:sym_pi_id}
\end{alignat}
Also, we have that, for any $g, h \in \mG$,
\begin{alignat}{1}
 \InsertPDF{figures/sym_pi_right.pdf} ~\raisebox{1.2em}{,}
 \label{eq:sym_pi_right}
\end{alignat}
which means $\tau_g \circ \tau_h \eqlocal \tau_{gh}$.
Thus, from Eq.~\eqref{eq:fca_eq} with $A = C$, we have
\begin{alignat}{1}
 \InsertPDF{figures/sym_pi.pdf} ~\raisebox{.5em}{.}
 \label{eq:sym_pi}
\end{alignat}
Clearly, for any $g \in \mG$, $\tau_g$ is deterministic and
$\tau_{g^{-1}}$ is the inverse of $\tau_g$.
Note that we also have%
\footnote{
We can also identify $\mI_M$ with $\StNP_C = \{ \cket{m} \}_{m \in \mI_M}$;
in this case, Eq.~\eqref{eq:sym_pi_left} implies that
$\tau$ can be interpreted as a left group action of $\mG$ on $\mI_M$.
A map $\psi: \mG \ni g \mapsto \psi_g \in \Aut(\mZ)$ is called
a \termdef{left group action} of $\mG$ on $\mZ$
if $\psi_{gh}(z) = \psi_g[\psi_h(z)]$ $~(\forall g,h \in \mG)$ and $\psi_1(z) = z$ hold
for any $z \in \mZ$.}
\begin{alignat}{1}
 \InsertPDF{figures/sym_pi_left.pdf} ~\raisebox{1.2em}{.}
 \label{eq:sym_pi_left}
\end{alignat}
We can easily verify
\begin{alignat}{1}
 \InsertPDF{figures/sym_pig.pdf} ~\raisebox{.5em}{.}
 \label{eq:sym_pig}
\end{alignat}
Indeed, from Eq.~\eqref{eq:classical_cup}, we have that,
for any $m,n \in \mI_M$,
\begin{alignat}{1}
 \InsertPDF{figures/sym_pig_proof.pdf} ~\raisebox{.5em}{.}
 \label{eq:sym_pig_proof}
\end{alignat}
Since any state in $C \otimes C$ is in the form
$\sum_{m=1}^M \sum_{n=1}^M c_{m,n} \cket{m} \otimes \cket{n}$
with $c_{m,n} \in \Real_+$,
Eq.~\eqref{eq:sym_pig_proof} yields Eq.~\eqref{eq:sym_pig}.
Similarly, we have
\begin{alignat}{1}
 \InsertPDF{figures/sym_pig2.pdf} ~\raisebox{.5em}{.}
 \label{eq:sym_pig2}
\end{alignat}

The second example is a group action, $\ol{\pi}:\mG \to \Aut(\Measgen_A)$,
of $\mG$ on $\Measgen_A$, where $\Measgen_A$ is a convex subset of $\Meas_A$.
Each automorphism of $\Measgen_A$ is in $\Vec_{A \to A}$,
i.e., each element of $\Aut(\Measgen_A)$ is linear.
We will say that $\Measgen_A \subseteq \Meas_A$ is
\termdef{symmetric under permutations of the measurement outcomes}
if, for any $h \in \Aut(\mI_M)$ and $e \in \Measgen_A$,
$h \circ e \in \Measgen_A$ holds.
This means that a process that first performs the measurement $e \in \Measgen_A$ and then
makes permutations among the measurement results is also in $\Measgen_A$.
$\ol{\pi}$ satisfies
\begin{alignat}{1}
 \InsertPDF{figures/sym_tau_id.pdf}
 \label{eq:sym_tau_id}
\end{alignat}
and
\begin{alignat}{1}
 \InsertPDF{figures/sym_tau.pdf}
 \label{eq:sym_tau}
\end{alignat}
for any $g,h \in \mG$.
Clearly, $\ol{\pi}_{g^{-1}}$ is the inverse of $\ol{\pi}_g$.
Each $\ol{f} \in \Aut(\Measgen_A)$ is reversible.

An example of $\Measgen_A$ is $\Meas_A$ itself; $\Meas_A$ is obviously
symmetric under permutations of the measurement outcomes.
It is also easily seen that the followings are equivalent:
\begin{enumerate}
 \item $\ol{f} \in \Aut(\Meas_A)$.
 \item $\ol{f} \in \Vec_{A \to A}$ is reversible and satisfies
       $e \circ \ol{f} \in \Meas_A$ for any $e \in \Meas_A$.
 \item $\ol{f} \in \Vec_{A \to A}$ is reversible, deterministic, and positive for effects,
       where we will call $\ol{g} \in \Vec_{A \to B}$ \termdef{positive for effects}
       if $\cra{e} \circ \ol{g} \in \Eff_A$ holds for any $\cra{e} \in \Eff_B$.
\end{enumerate}
Other examples of $\Measgen_A$ will be shown in Subsec.~\ref{subsec:opt_restricted}.
Note that $\Aut(\Measgen_A) \subseteq \Aut(\Meas_A)$ does not hold in general;
an example will be given in the example of quantum theory stated
in Subsec.~\ref{subsec:opt_restricted}.

\subsection{Symmetric properties}

We consider a set $(\mG, \tau, \ol{\pi})$, where $\tau:\mG \to \Aut(\mI_M)$ and
$\ol{\pi}:\mG \to \Aut(\Measgen_A)$ are group actions with a group $\mG$.
We will say that
a state preparation $\rho \coloneqq \sum_{n=1}^M \cket{\rho_n} \circ \cra{n} \in \Prep_A$
is \termdef{$(\mG,\tau,\ol{\pi})$-covariant} (or simply \termdef{$\mG$-covariant}) if
\begin{alignat}{1}
 \InsertPDF{figures/sym_rho.pdf}
 \label{eq:sym_rho}
\end{alignat}
holds for any $g \in \mG$.
We will also say that
a measurement $e^\sym \coloneqq \sum_{n=1}^M \cket{n} \circ \cra{e^\sym_n} \in \Measgen_A$ is
\termdef{$(\mG,\tau,\ol{\pi})$-covariant} if
\begin{alignat}{1}
 \InsertPDF{figures/sym_e.pdf}
 \label{eq:sym_e}
\end{alignat}
holds for any $g \in \mG$.
Equation~\eqref{eq:sym_rho} is the same as
\begin{alignat}{1}
 \InsertPDF{figures/sym_rho2.pdf}
 \label{eq:sym_rho2}
\end{alignat}
since
\begin{alignat}{1}
 \begin{overpic}{figures/sym_rho2_proof.pdf}
  \put(13,30){\footnotesize\eqref{eq:st}}
  \put(34,30){\footnotesize\eqref{eq:classical_cup_st_e2}}
 \end{overpic}
 \label{eq:sym_rho2_proof}
\end{alignat}
and
\begin{alignat}{1}
 \begin{overpic}{figures/sym_rho2_proof2.pdf}
  \put(21,37){\footnotesize\eqref{eq:st}}
  \put(52,44){\footnotesize\eqref{eq:classical_cup_st_e2}}
  \put(52,37){\footnotesize\eqref{eq:sym_pig2}}
 \end{overpic}
 ~\raisebox{1.0em}{.}
 \label{eq:sym_rho2_proof2}
\end{alignat}
Note that Eq.~\eqref{eq:sym_rho} is also the same as
\begin{alignat}{1}
 \InsertPDF{figures/sym_rho0.pdf} ~\raisebox{.5em}{.}
 \label{eq:sym_rho0}
\end{alignat}
Similar results hold for $e$; for example, Eq.~\eqref{eq:sym_e} is the same as
\begin{alignat}{1}
 \InsertPDF{figures/sym_e2.pdf}
 \label{eq:sym_e2}
\end{alignat}
since
\begin{alignat}{1}
 \begin{overpic}{figures/sym_e2_proof.pdf}
  \put(13,30){\footnotesize\eqref{eq:meas}}
  \put(34,30){\footnotesize\eqref{eq:classical_cup_st_e}}
 \end{overpic}
 \label{eq:sym_e2_proof}
\end{alignat}
and
\begin{alignat}{1}
 \begin{overpic}{figures/sym_e2_proof2.pdf}
  \put(22,33){\footnotesize\eqref{eq:meas}}
  \put(53,33){\footnotesize\eqref{eq:classical_cup_st_e}}
 \end{overpic}
 ~\raisebox{1.0em}{.}
 \label{eq:sym_e2_proof2}
\end{alignat}

\begin{thm}{}{Symmetry}
 Let $\Measgen_A$ be a convex subset of $\Meas_A$ that is
 symmetric under permutations of the measurement outcomes.
 Let $\tau:\mG \to \Aut(\mI_M)$ and $\ol{\pi}:\mG \to \Aut(\Measgen_A)$
 be group actions with a group $\mG$.
 If a state preparation $\rho \in \Prep_A$ is $(\mG,\tau,\ol{\pi})$-covariant,
 then, for any measurement $e \in \Measgen_A$,
 there exists a $(\mG,\tau,\ol{\pi})$-covariant measurement $e^\sym \in \Measgen_A$ satisfying
 $\PS(e^\sym) = \PS(e)$.
\end{thm}
\begin{proof}
 Let
 \begin{alignat}{1}
  \InsertPDF{figures/sym_esym.pdf} ~\raisebox{.5em}{.}
  \label{eq:sym_esym}
 \end{alignat}
 From the definition of $\Measgen_A$, we can easily verify $e^\sym \in \Measgen_A$.
 Equation~\eqref{eq:sym_esym} yields
 \begin{alignat}{1}
  \begin{overpic}{figures/sym_esym2.pdf}
   \put(29,34){\footnotesize\eqref{eq:sym_esym}}
   \put(29,28){\footnotesize\eqref{eq:sym_pig}}
  \end{overpic}
  ~\raisebox{.5em}{,}
  \label{eq:sym_esym2}
 \end{alignat}
 and thus
 \begin{alignat}{1}
  \begin{overpic}{figures/sym_esym_tau.pdf}
   \put(25,61){\footnotesize\eqref{eq:sym_esym2}}
   \put(25,23){\footnotesize\eqref{eq:sym_pi}}
   \put(25,19){\footnotesize\eqref{eq:sym_tau}}
   \put(73,19){\footnotesize\eqref{eq:sym_esym2}}
  \end{overpic}
 \end{alignat}
 holds for any $g \in \mG$, where the last equality follows from
 $\{ h : h \in \mG \} = \mG = \{ hg : h \in \mG \}$.
 Thus, $e^\sym$ is $(\mG,\tau,\ol{\pi})$-covariant.
 Also, we obtain
 \begin{alignat}{1}
  \begin{overpic}{figures/sym_PS.pdf}
   \put(20.5,27){\footnotesize\eqref{eq:sym_esym2}}
   \put(71.5,27){\footnotesize\eqref{eq:sym_rho}}
  \end{overpic}
  ~\raisebox{.5em}{,}
 \end{alignat}
 which gives $\PS(e^\sym) = \PS(e)$.
\end{proof}

Theorem~\ref{thm:Symmetry} immediately yields the following corollary.
\begin{cor}{}{Symmetry}
 Let $\Measgen_A$ be a convex subset of $\Meas_A$ that is
 symmetric under permutations of the measurement outcomes.
 Let $\tau:\mG \to \Aut(\mI_M)$ and $\ol{\pi}:\mG \to \Aut(\Measgen_A)$
 be group actions with a group $\mG$.
 If a state preparation $\rho \in \Prep_A$ is $(\mG,\tau,\ol{\pi})$-covariant,
 then there exists a $(\mG,\tau,\ol{\pi})$-covariant measurement that is optimal
 for Problem~\eqref{eq:main_gen}.
\end{cor}

\begin{ex}[quantum theory]
 We consider the special case of $\Measgen_A = \Meas_A$.
 Any $\ol{f} \in \Aut(\Meas_A)$ is expressed in the form
 \begin{alignat}{1}
  \ol{f} \circ \cket{\rho} = U_{\ol{f}} \cdot \cket{\rho} \cdot U_{\ol{f}}^\dagger,
  \label{eq:quant_f_unitary}
 \end{alignat}
 where $U_{\ol{f}}$ is a unitary or anti-unitary matrix of order $\NA$
 and $^\dagger$ denotes the conjugate transpose%
 \footnote{Easy proof: One can easily verify that any $\ol{f} \in \Aut(\Meas_A)$
 maps pure effects to pure effects and thus
 maps normalized pure states to normalized pure states.
 Therefore, according to Lemma~4 of Ref.~\cite{Fri-Li-Poo-Sze-2011},
 $\ol{f}$ is expressed in the form of Eq.~\eqref{eq:quant_f_unitary}
 or in the form $\ol{f} \circ \cket{\rho} = [\Tr~\cket{\rho}] \cdot \cket{\phi}$
 with a fixed $\cket{\phi} \in \StNP_A$.
 The latter case is ruled out since $\ol{f}$ is reversible.}.
 Thus, $\ol{\pi}_g \in \Aut(\Meas_A)$ must be in the form
 \begin{alignat}{1}
  \ol{\pi}_g \circ \cket{\rho_m} &= U_g \cdot \cket{\rho_m} \cdot U_g^\dagger,
 \end{alignat}
 where $U_g$ is a unitary or anti-unitary matrix of order $\NA$.
 Since $\ol{\pi}:\mG \to \Aut(\Meas_A)$ is a group action,
 $U_1 \cdot \cket{\rho_m} \cdot U_1^\dagger = \ident_\NA$ and
 $U_g [ U_h \cdot \cket{\rho_m} \cdot U_h^\dagger] U_g^\dagger
 = U_{gh} \cdot \cket{\rho_m} \cdot U_{gh}^\dagger$ must hold for any $g,h \in \mG$.
 This type of symmetry has been discussed in Ref.~\cite{Nak-Usu-2013-group}.
\end{ex}

\subsection{Optimization over a restricted class of measurements} \label{subsec:opt_restricted}

Theorem~\ref{thm:Symmetry} can be utilized to
optimize over a restricted class of measurements,
as we will see in this subsection.
We here consider state discrimination problems in a bipartite system.

Let us introduce a three-party: Alice, Bob, and Charlie.
We will consider the following scenario.
Charlie randomly chooses one of the $M$ states
$\rho_1^\N,\ldots,\rho_M^\N \in \StN_\AB$
with prior probabilities $\xi_1,\ldots,\xi_M$,
which is expressed by the state preparation
$\rho \coloneqq \sum_{n=1}^M \cket{\rho_n} \circ \cra{n} \in \Prep_\AB$
with $\cket{\rho_n} \coloneqq \xi_n \cket{\rho_n^\N}$ and $\cra{n} \in \Eff_C$.
$A$ and $B$ respectively refer to the systems of Alice and Bob.
Also, $C$ is a classical system with $\NC = M$.
Charlie sends the state to Alice and Bob,
who perform a measurement 
$e \coloneqq \sum_{m=1}^M \cket{m} \circ \cra{e_m} \in \Meas_\AB$ with $\cra{e_m} \in \Eff_\AB$.
Then, Alice, Bob, and Charlie check whether Alice and Bob correctly determine the state.
Similarly to Eq.~\eqref{eq:PS}, the average success probability of a measurement
$e$ is depicted as
\begin{alignat}{1}
 \InsertPDF{figures/seq_PS.pdf}
 ~\raisebox{1.5em}{.}
 \label{eq:seq_PS}
\end{alignat}

If Alice and Bob can perform any measurement,
this problem is expressed by Eq.~\eqref{eq:main} with $\AB$ instead of $A$.
In this subsection, assume that they can only perform restricted measurements.
We consider four classes of measurements: sequential, LOCC, separable, and PT.
Let $D$ and $D'$ be classical systems, which can be infinite-dimensional.

A measurement $e \in \Meas_\AB$ is referred to as \termdef{sequential} if
it can be expressed in the form
\begin{alignat}{1}
 \InsertPDF{figures/meas_seq.pdf}
 \label{eq:meas_seq}
\end{alignat}
with $a \in \Meas_A^D$ and $b \in \Meas_{D \otimes B}$.
$b \in \Meas_{D \otimes B}$ can be interpreted as a classical controlled measurement.
Indeed, let us define $b^i$ as a measurement of $B$ depicted by
\begin{alignat}{1}
 \InsertPDF{figures/measBm.pdf}~\raisebox{2.0em}{;}
 \label{eq:measBm}
\end{alignat}
then, we have
\begin{alignat}{1}
 \begin{overpic}{figures/measB.pdf}
  \put(17,22){\footnotesize\eqref{eq:classical_id}}
  \put(58,22){\footnotesize\eqref{eq:measBm}}
 \end{overpic}
 ~\raisebox{.5em}{.}
 \label{eq:measB}
\end{alignat}
$b$ can be interpreted as a process which performs a measurement $b^i \in \Meas_B$
when the state $\cket{i}$ is inputted to the system $D$.

A measurement $e \in \Meas_\AB$ is referred to as \termdef{LOCC}
if it can be expressed in the form
\begin{alignat}{1}
 \InsertPDF{figures/meas_LOCC.pdf}
 \label{eq:meas_LOCC}
\end{alignat}
with $a_1 \in \ProcD_{A \to A_1 \otimes D}$,
$a_k \in \ProcD_{A_{k-1} \otimes D \to A_k \otimes D}$ $~(k \in \{ 2,\ldots,n-1 \})$,
$a_n \in \Meas_{A_{n-1} \otimes D}^D$,
$b_1 \in \ProcD_{D \otimes B \to D \otimes B_1}$,
$b_l \in \ProcD_{D \otimes B_{l-1} \to D \otimes B_l}$ $~(l \in \{ 2,\ldots,n-1 \})$,
and $b_n \in \Meas_{D \otimes B_{n-1}}$,
where $n$ is some natural number.

A measurement $e \in \Meas_\AB$ is referred to as \termdef{separable} if
it can be expressed in the form
\begin{alignat}{1}
 \InsertPDF{figures/meas_separable.pdf} ~\raisebox{.5em}{,}
 \label{eq:meas_separable}
\end{alignat}
where $a \in \Proc_{A \to D}$, $b \in \Proc_{B \to D'}$, and
$c \in \Proc_{D \otimes D' \to C}$.
Any $c \in \Proc_{D \otimes D' \to C}$ can be expressed by
\begin{alignat}{1}
 \begin{overpic}{figures/meas_separable_c.pdf}
  \put(13,28){\footnotesize\eqref{eq:classical_id}}
 \end{overpic}
 ~\raisebox{1.5em}{,}
 \label{eq:meas_separable_c}
\end{alignat}
where $c_{i,j,m} \in \Real_+$ is the scalar enclosed by the auxiliary box.
Using Eq.~\eqref{eq:meas_separable_c}, we can easily see that
a necessary and sufficient condition for $e \in \Meas_\AB$ to be separable
is that $\cra{m} \circ e$ is separable for each $m \in \mI_M$.

A measurement $e \in \Meas_\AB$ is referred to as \termdef{PT} if
\begin{alignat}{1}
 \InsertPDF{figures/meas_PT.pdf}
 \label{eq:meas_PT}
\end{alignat}
holds for any system $A'$ and any deterministic extended process $\ol{f} \in \Vec_{A' \to A}$
that is positive for effects.

Let $\MeasseqAB$ be the set of all sequential measurements from $A$ to $B$.
Also, let $\MeasLOCC_\AB$, $\MeasSEP_\AB$, and $\MeasPT_\AB$ be, respectively,
a set of all LOCC, separable, and PT measurements of $\AB$.

\begin{proposition}{}{Meas}
 \begin{alignat}{3}
  \MeasseqAB &\subseteq \MeasLOCC_\AB &&\subseteq \MeasSEP_\AB \nonumber \\
  & &&\subseteq \MeasPT_\AB && \subseteq \Meas_\AB.
 \end{alignat}
\end{proposition}
\begin{proof}
 $\MeasseqAB \subseteq \MeasLOCC_\AB$ and
 $\MeasPT_\AB \subseteq \Meas_\AB$ obviously hold.
 $\Measseq_\AB \subseteq \MeasSEP_\AB$ follows from
 \begin{alignat}{1}
  \begin{overpic}{figures/meas_separable_seq.pdf}
   \put(23,49){\footnotesize\eqref{eq:classical_id}}
   \put(57,49){\footnotesize\eqref{eq:classical_cup_st_e2}}
  \end{overpic}
  ~\raisebox{.5em}{,}
  \label{eq:meas_separable_seq}
 \end{alignat}
 where $c$ and $b'$ are the processes enclosed by the upper and lower auxiliary boxes, respectively.
 $\MeasLOCC_\AB \subseteq \MeasSEP_\AB$ can be immediately proved in the same way.
 The proof is completed by showing $\MeasSEP_\AB \subseteq \MeasPT_\AB$.
 Let $e \in \MeasSEP_\AB$ be expressed in the form of Eq.~\eqref{eq:meas_separable}.
 Arbitrarily choose a deterministic extended process $\ol{f} \in \Vec_{A' \to A}$
 that is positive for effects.
 Since $e' \coloneqq e \circ (\ol{f} \otimes \id_B)$ is obviously deterministic,
 it remains to prove that $e'$ is a process.
 We have
 \begin{alignat}{1}
  \begin{overpic}{figures/meas_PT_separable.pdf}
   \put(21,23){\footnotesize\eqref{eq:meas_separable}}
   \put(50,23){\footnotesize\eqref{eq:meas_separable_c}}
  \end{overpic}
  ~\raisebox{.5em}{.}
  \label{eq:meas_PT_separable}
 \end{alignat}
 Since $\ol{f}$ is positive for effects,
 the extended effect enclosed by the auxiliary box is an effect,
 and thus $e'$ is a process.
 Therefore, $e$ is PT.
\end{proof}

\begin{proposition}{}{QuantPT}
 In quantum theory, $\MeasSEP_\AB = \MeasPT_\AB$ holds.
\end{proposition}
\begin{proof}
 Since $\MeasSEP_\AB \subseteq \MeasPT_\AB$ holds,
 it suffices to show $\MeasSEP_\AB \supseteq \MeasPT_\AB$,
 i.e., for any $e \in \MeasPT_\AB$ and $m \in \mI_M$,
 $\cra{m} \circ e$ is separable.

 Assume, by contradiction, that $\cra{e_m} \coloneqq \cra{m} \circ e$ is
 entangled for some $e \in \MeasPT_\AB$ and $m \in \mI_M$.
 It has been shown in Ref.~\cite{Hor-Hor-Hor-1996} that
 there exists an extended state $\cket{\ol{v}} \in \Vec_\AB$
 such that $\craket{e_m|\ol{v}} < 0$ and
 $[\cra{b} \otimes \cra{b'}] \cket{\ol{v}} \ge 0$ hold for
 any $\cra{b} \in \Eff_A$ and $\cra{b'} \in \Eff_B$.
 One can easily verify that, for each $\cra{b} \in \Eff_A$,
 $\cket{v'} \coloneqq [\cra{b} \otimes \id_B] \circ \cket{\ol{v}} \in \St_B$ holds
 since $\craket{b'|v'} \ge 0$ holds for any $\cra{b'} \in \Eff_B$.
 Let $\cket{v_B} \coloneqq [\cra{\gdis_A} \otimes \id_B] \circ \cket{\ol{v}} \in \St_B$.
 Assume, without loss of generality, that
 $\cket{v_B}$ is full rank
 (if not, we can replace $\cket{\ol{v}}$ with
 $\cket{\ol{v}} + c \cket{\rho_0} \otimes \cket{\gdisT_B}$
 with $0 < c \in \Real_+$, $\cket{\rho_0} \in \StN_A$, and
 $\cket{\gdisT_B} \coloneqq \ident_\NB$, where $c$ is sufficiently small
 such that $\craket{e_m|\ol{v}} < 0$).
 It is easily seen that there exists a reversible process $g \in \Proc_{B \to B}$
 such that $g \circ \cket{v_B} = \cket{\gdisT_B}$.
 Let $g^{-1}$ be the inverse of $g$.
 It is well-known that, for each system $A$, there exist $\cket{\cup_A} \in \St_{A \otimes A}$
 and $\cra{\cap_A} \in \Eff_{A \otimes A}$ such that
 $[\cra{\cap_A} \otimes \id_A] \circ [\id_A \otimes \cket{\cup_A}] = \id_A$
 and $\cra{\cap_A} \circ [\cket{\gdisT_A} \otimes \id_A] = \cra{\gdis_A}$
 (see, e.g., \cite{Coe-Kis-2017}).
 Let
 \begin{alignat}{1}
  \InsertPDF{figures/meas_PT_quantum_f.pdf} ~\raisebox{.5em}{.}
  \label{eq:meas_PT_quantum_f}
 \end{alignat}
 One can easily verify
 \begin{alignat}{1}
  \InsertPDF{figures/meas_PT_quantum2.pdf} ~\raisebox{2.0em}{,}
  \label{eq:meas_PT_quantum2}
 \end{alignat}
 i.e., $\ol{f}$ is deterministic.
 Also, $\ol{f}$ is positive for effects since,
 for any $\cra{b} \in \Eff_A$,
 from $[\cra{b} \otimes \id_B] \circ \cket{\ol{v}} \in \St_B$,
 $\cra{b} \circ \ol{f} \in \Eff_B$ holds.
 Moreover, we have
 \begin{alignat}{1}
  \InsertPDF{figures/meas_PT_quantum.pdf} ~\raisebox{2.0em}{,}
  \label{eq:meas_PT_quantum}
 \end{alignat}
 which means that $\cra{e_m} \circ (\ol{f} \otimes \id_B)$ is not an effect.
 This contradicts $e \in \MeasPT_\AB$.
 Therefore, $\MeasSEP_\AB \supseteq \MeasPT_\AB$ holds.
\end{proof}

Obviously, $\MeasseqAB$, $\MeasLOCC_\AB$, $\MeasSEP_\AB$, and $\MeasPT_\AB$ are
convex subsets of $\Meas_\AB$ that are symmetric under permutations of the measurement outcomes.
Thus, the following corollary follows immediately from Theorem~\ref{thm:Symmetry}
with $\AB$ instead of $A$:
\begin{cor}{}{BiSymmetry}
 Let $\Measgen_\AB$ be $\MeasseqAB$, $\MeasLOCC_\AB$, $\MeasSEP_\AB$, or $\MeasPT_\AB$.
 Let $\tau:\mG \to \Aut(\mI_M)$ and $\ol{\pi}:\mG \to \Aut(\Measgen_A)$
 be group actions with a group $\mG$.
 If a state preparation $\rho \in \Prep_\AB$ is $(\mG,\tau,\ol{\pi})$-covariant,
 then, for any measurement $e \in \Measgen_\AB$,
 there exists a $(\mG,\tau,\ol{\pi})$-covariant measurement
 $e^\sym \in \Measgen_\AB$ that satisfies $\PS(e^\sym) = \PS(e)$.
\end{cor}

\begin{ex}[quantum theory]
 It is easily seen that each element of $\Aut(\Measgen_\AB)$
 with $\Measgen_\AB \in \{ \MeasseqAB, \MeasLOCC_\AB, \MeasSEP_\AB \}$
 maps normalized pure effects to normalized pure effects
 (where a pure effect is called normalized if its trace is one)
 and thus maps normalized pure states to normalized pure states.
 Let $\mP$ be the set of all normalized separable pure states of $\AB$.
 Also, let $\mP^*$ be the set of all normalized separable pure effects of $\AB$.
 For each $\cra{e} \in \mP^*$,
 there exists a sequential measurement $\{ \cra{e_m} \}_{m \in \mI_M}$ satisfying
 $\cra{e_1} = \cra{e}$.
 Thus, we can easily see that, for any $\ol{f} \in \Aut(\MeasseqAB)$,
 we have $\{ \cra{e} \circ \ol{f} : \cra{e} \in \mP^* \} = \mP^*$,
 i.e.,
 \begin{alignat}{1}
  \{ \ol{f} \circ \cket{\psi} : \cket{\psi} \in \mP \} = \mP.
  \label{eq:quant_f_P}
 \end{alignat}
 It follows, by the same argument, that Eq.~\eqref{eq:quant_f_P} also holds
 for any $\ol{f} \in \Aut(\MeasLOCC_\AB)$ and for any $\ol{f} \in \Aut(\MeasSEP_\AB)$.
 According to Theorem~3 of Ref.~\cite{Fri-Li-Poo-Sze-2011},
 $\ol{f} \in \Vec_{\AB \to \AB}$ satisfies Eq.~\eqref{eq:quant_f_P}
 if and only if $\ol{f}$ has one of the following forms:
 \begin{enumerate}
  \item $\ol{f} \circ \cket{\rho} = (U_1 \otimes U_2) \cdot \cket{\rho} \cdot (U_1 \otimes U_2)^\dagger$,
  \item $\ol{f} \circ \cket{\rho} = (U_1 \otimes U_2) \cdot
        [\cross_{A,B} \circ \cket{\rho}] \cdot (U_1 \otimes U_2)^\dagger$,
 \end{enumerate}
 where $\cket{\rho} \in \St_\AB$.
 $U_1$ and $U_2$ are, respectively, unitary or anti-unitary matrices of
 order $\NA$ and $\NB$.
 $\ol{f}$ can have the form of 2) only in the case of $A \cong B$.
 Note that $\cross_{A,B}$ satisfies $\cross_{A,B} \circ [\cket{\rho_1} \otimes \cket{\rho_2}]
 = \cket{\rho_2} \otimes \cket{\rho_1}$ for any $\cket{\rho_1} \in \St_A$ and $\cket{\rho_2} \in \St_B$.
 In each of these cases, we have:
 \begin{enumerate}
  \item $\cra{e} \circ \ol{f} = (U_1 \otimes U_2)^\dagger \cdot \cra{e} \cdot (U_1 \otimes U_2)$,
  \item $\cra{e} \circ \ol{f} = (U_2 \otimes U_1)^\dagger \cdot
        [\cra{e} \circ \cross_{B,A}] \cdot (U_2 \otimes U_1)$,
 \end{enumerate}
 where $\cra{e} \in \Eff_\AB$.

 One can easily see that any $\ol{f} \in \Aut(\MeasseqAB)$ must have the form of 1)
 unless $A \cong B \cong I$.
 (Indeed, if there exists $\ol{f} \in \Aut(\MeasseqAB)$ with the form of 2),
 then, $e \circ \ol{f} \in \Measseq_{B \to A}$ holds for any $e \in \MeasseqAB$.
 In this case, since $\{ e \circ \ol{f} : e \in \MeasseqAB \} = \MeasseqAB$ holds
 from $\ol{f} \in \Aut(\MeasseqAB)$,
 $\MeasseqAB \subseteq \Measseq_{B \to A}$ must hold.
 However, it is immediately seen that there exists $e \in \MeasseqAB$
 satisfying $e \not\in \Measseq_{B \to A}$.)
 This type of symmetry with respect to $\MeasseqAB$ has been discussed
 in Ref.~\cite{Nak-Kat-Usu-2018-seq_gen}.
 If $\Measgen_\AB \in \{ \MeasLOCC_\AB, \MeasSEP_\AB \}$ and
 $A \cong B$ hold, then $\ol{f} \in \Aut(\Measgen_\AB)$ can have the form of 2).
\end{ex}

\section{Conclusion}

The problem of discrimination of symmetric states in an OPT has been investigated
in diagrammatic terms.
It is well-known that, in quantum theory, if states have a certain symmetry, then
there exists a minimum-error measurement that has the same type of symmetry.
We showed in Theorem~\ref{thm:Symmetry} that this property is also valid in a more general OPT.
We also showed that this result can be utilized to optimize over a restricted class
of measurements.
Although we discuss only the minimum-error strategy to simplify the discussion,
this result can be easily applied to other various criteria,
such as the Bayes criterion or the minimax criterion.

\section*{Acknowledgment}

I am grateful to O.~Hirota, K.~Kato, and T.~S.~Usuda for support.
This work was supported by JSPS KAKENHI Grant Number JP19K03658.

\bibliographystyle{myieeetr}

\input{bibliography.bbl}

\end{document}

%% file: settings_quant.tex
\usepackage{bm}

\theoremstyle{definition}

\tcbuselibrary{theorems}
\tcbset{
  code={},
  colback=white,
  noparskip,
  before=\par\bigskip\noindent,
  coltitle=black,fonttitle=\upshape\bfseries,
  breakable,frame empty,interior empty,colframe=white,
  description delimiters parenthesis,
  enhanced,theorem style=plain,
  boxsep=0mm,
  leftrule=0mm,rightrule=0mm,
  left=0mm,right=0mm,top=0mm,bottom=0mm,
  before skip=10pt,after skip=10pt,
  posstyle/.style={},
  thmstyle/.style={},
}
\newtcbtheorem{postulate}{Postulate}{posstyle}{pos}
\newtcbtheorem{thm}{Theorem}{thmstyle}{thm}
\newtcbtheorem[use counter from=thm]{lemma}{Lemma}{thmstyle}{lemma}
\newtcbtheorem[use counter from=thm]{proposition}{Proposition}{thmstyle}{pro}
\newtcbtheorem[use counter from=thm]{cor}{Corollary}{thmstyle}{cor}

\newcounter{ex}

\ExplSyntaxOn
\NewDocumentEnvironment{ex}{o}
 {
  \refstepcounter{ex}
  \par\smallskip
  \noindent
  \textbf{\IfNoValueTF{#1}{Example}{Example~of~#1}~}
 }
 {\par\smallskip
 }
\ExplSyntaxOff

\newcommand{\mG}{\mathcal{G}}

\newcommand{\mI}{\mathcal{I}}

\newcommand{\mM}{\mathcal{M}}

\newcommand{\mP}{\mathcal{P}}

\newcommand{\mS}{\mathcal{S}}

\newcommand{\mZ}{\mathcal{Z}}

\newcommand{\PS}{P_{\rm S}}

\newcommand{\ident}{\hat{1}}

\newcommand{\Real}{\mathbf{R}}
\newcommand{\Complex}{\mathbf{C}}

\newcommand{\Tr}{{\rm Tr}}

\def\gauss_sym#1{{\lfloor #1 \rfloor}}